\documentclass[conference,a4paper]{APSIPA2018}
\usepackage{multirow}
\usepackage{graphicx}
\usepackage{amsmath}
\usepackage[psamsfonts]{amssymb}
\usepackage{amsxtra}
\usepackage{threeparttable}
\usepackage{cite}
\usepackage{booktabs}
\usepackage{siunitx}
\usepackage{url}

\makeatletter
\let\MYcaption\@makecaption

\usepackage[font+=footnotesize]{subcaption}
\let\@makecaption\MYcaption
\makeatother

\usepackage{mgmmathtool}

\newtheorem{theorem}{Theorem}

\begin{document}


\title{%
Independent Low-Rank Matrix Analysis
Based on Time-Variant Sub-Gaussian Source Model
}


\author{%
\authorblockN{%
Shinichi Mogami\authorrefmark{1},
Norihiro Takamune\authorrefmark{1},
Daichi Kitamura\authorrefmark{2},
Hiroshi Saruwatari\authorrefmark{1},\\
Yu Takahashi\authorrefmark{3},
Kazunobu Kondo\authorrefmark{3},
Hiroaki Nakajima\authorrefmark{3},
and Nobutaka Ono\authorrefmark{4}
}
\authorblockA{%
\authorrefmark{1}
The University of Tokyo,
Tokyo, Japan
}
\authorblockA{%
\authorrefmark{2}
National Institute of Technology, Kagawa College,
Kagawa, Japan
}
\authorblockA{%
\authorrefmark{3}
Yamaha Corporation,
Shizuoka, Japan
}
\authorblockA{%
\authorrefmark{4}
Tokyo Metropolitan University,
Tokyo, Japan
}
}

\maketitle
\thispagestyle{empty}


\begin{abstract}
  Independent low-rank matrix analysis (ILRMA)
  is a fast and stable method for blind audio source separation.
  Conventional ILRMAs assume time-variant (super-)Gaussian source models,
  which can only represent signals that follow a super-Gaussian distribution.
  In this paper,
  we focus on ILRMA based on a generalized Gaussian distribution (GGD-ILRMA)
  and propose a new type of GGD-ILRMA that adopts a time-variant sub-Gaussian distribution for the source model.
  By using a new update scheme called
  generalized iterative projection for homogeneous source models,
  we obtain a convergence-guaranteed update rule for demixing spatial parameters.
  In the experimental evaluation,
  we show the versatility of the proposed method,
  i.e., the proposed time-variant sub-Gaussian source model
  can be applied to various types of source signal.

\end{abstract}


\section{Introduction}

Blind source separation (BSS)%
~\cite{PComon1994_ICA,PSmaragdis1998_BSS,%
Sawada2004_FDICA,HSaruwatari2006_FDICA,%
AHiroe2006_IVA,TKim2007_IVA,NOno2011_AuxIVA,Yatabe2018_proxBSS}
is a technique
for extracting specific sources from an observed multichannel mixture signal
without knowing a priori information about the mixing system.
The most commonly used algorithm for BSS
in the (over)determined case ($\text{number of microphones}\geq\text{number of sources}$)
is independent component analysis~(ICA)~\cite{PComon1994_ICA}.
As a state-of-the-art ICA-based BSS method,
Kitamura et al.\ proposed
\emph{independent low-rank matrix analysis~(ILRMA)}%
~\cite{DKitamura2016_ILRMA,Kitamura2018_ILRMAbook},
which is a unification of
independent vector analysis~(IVA)~\cite{AHiroe2006_IVA,TKim2007_IVA}
and nonnegative matrix factorization~(NMF)~\cite{DDLee1999_NMF}.
ILRMA assumes both statistical independence between sources and a low-rank time-frequency structure for each source,
and the frequency-wise demixing systems are estimated
without encountering the permutation problem%
~\cite{Sawada2004_FDICA,HSaruwatari2006_FDICA}.
ILRMA is a faster and more stable algorithm than
multichannel NMF (MNMF)~\cite{AOzerov2010_MNMF,Ozerov2011_ISMNTF,HSawada2013_MNMF},
which is an algorithm for BSS
that estimates the mixing system on the basis of spatial covariance matrices.

The original ILRMA based on Itakura--Saito (IS) divergence
assumes a time-variant isotropic complex Gaussian distribution
for the source generative model.
Hereafter, we refer to the original ILRMA as \emph{IS-ILRMA}.
Recently, various types of source generative model
have been proposed in ILRMA for robust BSS.
In particular,
$t$-ILRMA~\cite{Mogami2017_tILRMA} and GGD-ILRMA~\cite{Kitamura2018_GGDILRMA,Ikeshita2018_GGDILRMA}
have been proposed as generalizations of IS-ILRMA
with a complex Student's $t$ distribution
and a complex generalized Gaussian distribution (GGD), respectively.
In $t$-ILRMA and GGD-ILRMA,
the kurtosis of the generative models' distributions
can be parametrically changed
along with
the degree-of-freedom parameter in Student's $t$ distribution
and
the shape parameter in the GGD.
By changing the kurtosis of the distributions,
we can control how often the source signal outputs outliers
or its expected sparsity.
In particular, in sub-Gaussian models,
i.e.,\ models that follow distributions with a \emph{platykurtic} shape,
the source signal rarely outputs outliers.
Therefore, the sub-Gaussian modeling of sources is expected to
accurately estimate the source spectrogram without ignoring
its important spectral peaks.
Furthermore, many audio sources follow platykurtic distributions;
it is known that
musical instrument signals obey
sub-Gaussian distributions~\cite{Naik2012_MusicDst}.


However, neither conventional $t$-ILRMA nor GGD-ILRMA
assumes that the source model follows
a sub-Gaussian distribution.
Both $t$-ILRMA and GGD-ILRMA can adopt only a super-Gaussian (or Gaussian) model,
i.e.,\ models that follow distributions with a \emph{leptokurtic} shape,
for the source model.
In $t$-ILRMA,
this is because the complex Student's $t$ distribution
becomes only super-Gaussian for any degree-of-freedom parameter.
In GGD-ILRMA, on the other hand, it is because
the estimation algorithm for the demixing matrix
has not yet been derived for a sub-Gaussian case,
although the GGD itself can represent a sub-Gaussian distribution
depending on its shape parameter.
More specifically, the conventional
\emph{iterative projection (IP)}~\cite{NOno2011_AuxIVA},
which is an algorithm that updates the demixing matrix mainly used in
IVA and ILRMA,
cannot be applied to sub-Gaussian-based GGD-ILRMA
owing to mathematical difficulties.


In this paper, we propose a new type of ILRMA
that assumes time-variant sub-Gaussian source models.
This paper includes three novelties.
First,
we construct a new update scheme for the demixing matrix
called \emph{generalized IP for homogeneous source models (GIP-HSM)}.
Second, we derive a convergence-guaranteed update rule for the demixing matrix
in GGD-ILRMA with a shape parameter of four.
To the best of our knowledge, this is the world's first attempt
to model the source signal with a time-variant sub-Gaussian distribution,
and we derive the update rule by applying the above-mentioned scheme.
Third, we show the validity of the proposed sub-Gaussian GGD-ILRMA
via BSS experiments on music and speech signals.
We confirm that the proposed method is a versatile approach
to source modeling,
i.e.,
the proposed time-variant sub-Gaussian model
can represent super-Gaussian or Gaussian signals
as well as sub-Gaussian signals owing to its time-variant nature,
whereas the conventional models can only represent
super-Gaussian or Gaussian signals.

\section{Problem Formulation}

\subsection{Formulation of Demixing Model}
Let $N$ and $M$ be the numbers of sources and channels, respectively.
The short-time Fourier transforms (STFTs)
of the multichannel source, observed, and estimated signals are defined as
\begin{align}
  \bm s_{ij}&=(s_{ij1},\ldots, s_{ijN})\T \in\cset^N,\\
  \bm x_{ij}&=(x_{ij1},\ldots, x_{ijM})\T \in\cset^M,\\
  \bm y_{ij}&=(y_{ij1},\ldots, y_{ijN})\T \in\cset^N,
\end{align}
where $i=1,\ldots, I$; $j=1,\ldots, J$; $n=1,\ldots, N$; and $m=1\ldots,M$
are the integral indices of the
frequency bins, time frames, sources, and channels, respectively,
and ${}\T$ denotes the transpose.
We assume the mixing system
\begin{align}
  \bm x_{ij} = \bm A_i \bm s_{ij},
\end{align}
where $\bm A_i = (\bm a_{i1}, \ldots, \bm a_{iN}) \in \cset^{M\times N}$
is a frequency-wise mixing matrix
and $\bm a_{in}$ is the steering vector for the $n$th source.
When $M=N$ and $\bm A_i$ is not a singular matrix,
the estimated signal $\bm y_{ij}$ can be expressed as
\begin{align}
  \bm y_{ij} = \bm W_i \bm x_{ij},
\end{align}
where $\bm W_i = \bm A_i^{-1} = (\bm w_{i1}, \ldots, \bm w_{iN})\Ht$
is the demixing matrix,
$\bm w_{in}$ is the demixing filter for the $n$th source,
and ${}\Ht$ denotes the Hermitian transpose.
ILRMA estimates both $\bm W_i$ and $\bm y_{ij}$
from only the observation $\bm x_{ij}$ assuming statistical independence
between $s_{ijn}$ and $s_{ijn'}$, where $n\ne n'$.

\subsection{Generative Model and Cost Function in GGD-ILRMA}
GGD-ILRMA utilizes the isotropic complex GGD.
The probability density function of the GGD is
\begin{align}
  p(z) = \dfrac{
    \beta
    }{
    2\pi r^2 \Gamma(2/\beta)
  }
  \exp\p[4]{-{\dfrac{\abs{z}^\beta}{r^\beta}}}
  ,\label{eq:pdf-GGD}
\end{align}
where $\beta$ is the shape parameter, $r$ is the scale parameter,
and $\Gamma(\cdot)$ is the gamma function.
Fig.~\ref{fig:ggd-shape} shows the shapes of the GGD with $\beta=2$ and $\beta=4$.
When $\beta=2$, \eqref{eq:pdf-GGD} corresponds to
the probability density function of the complex Gaussian distribution
with a \emph{mesokurtic} shape.
In the case of $0<\beta<2$, the distribution becomes
super-Gaussian with a {leptokurtic} shape.
In the case of $\beta>2$, the distribution becomes
sub-Gaussian with a {platykurtic} shape.

\begin{figure}[tp]
  \centering
  \begin{minipage}{.49\linewidth}
    \centering
    \includegraphics[width=\linewidth]{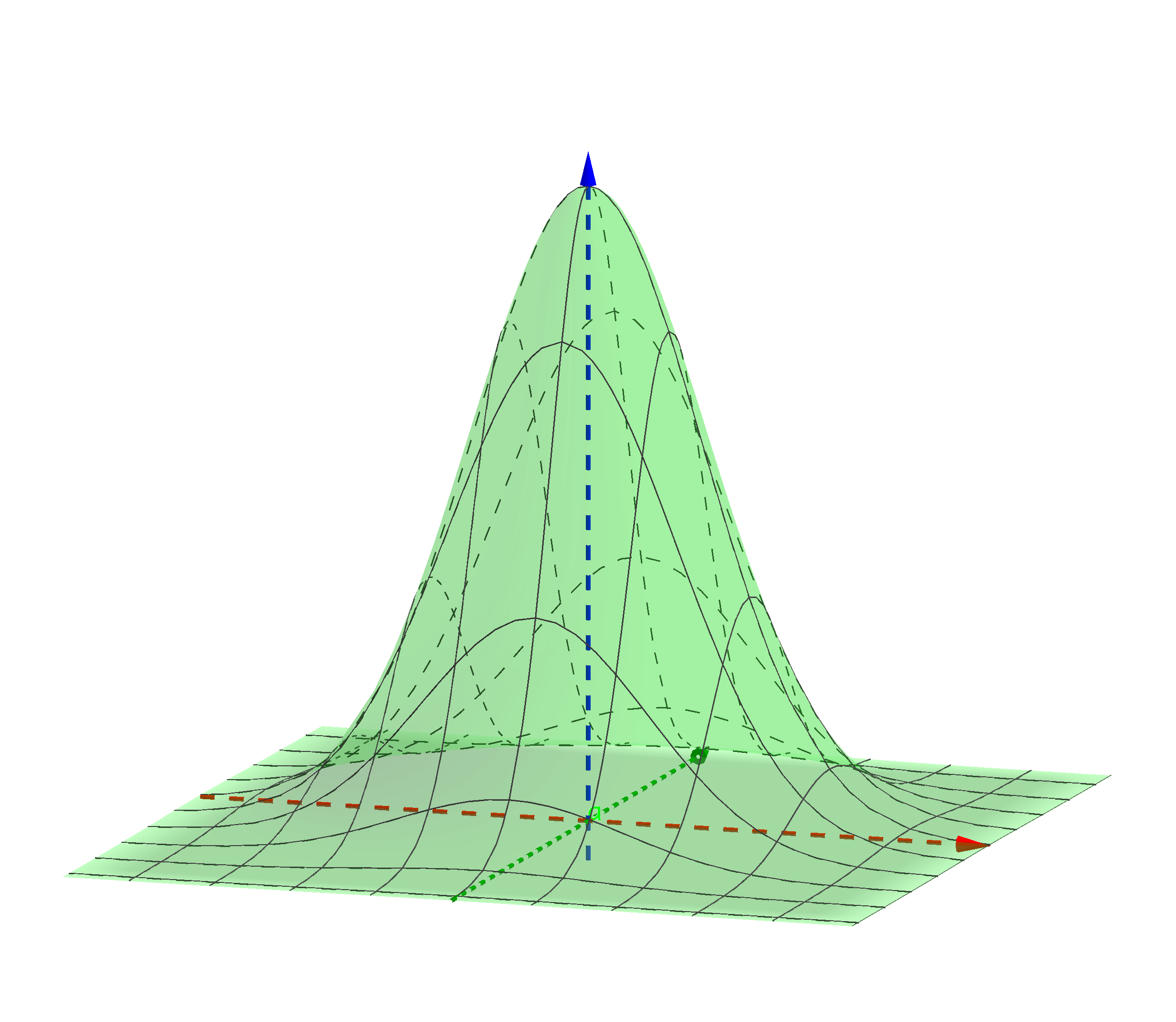}
    \subcaption{GGD with $\beta=2$.}
  \end{minipage}%
  \begin{minipage}{.49\linewidth}
    \centering
    \includegraphics[width=\linewidth]{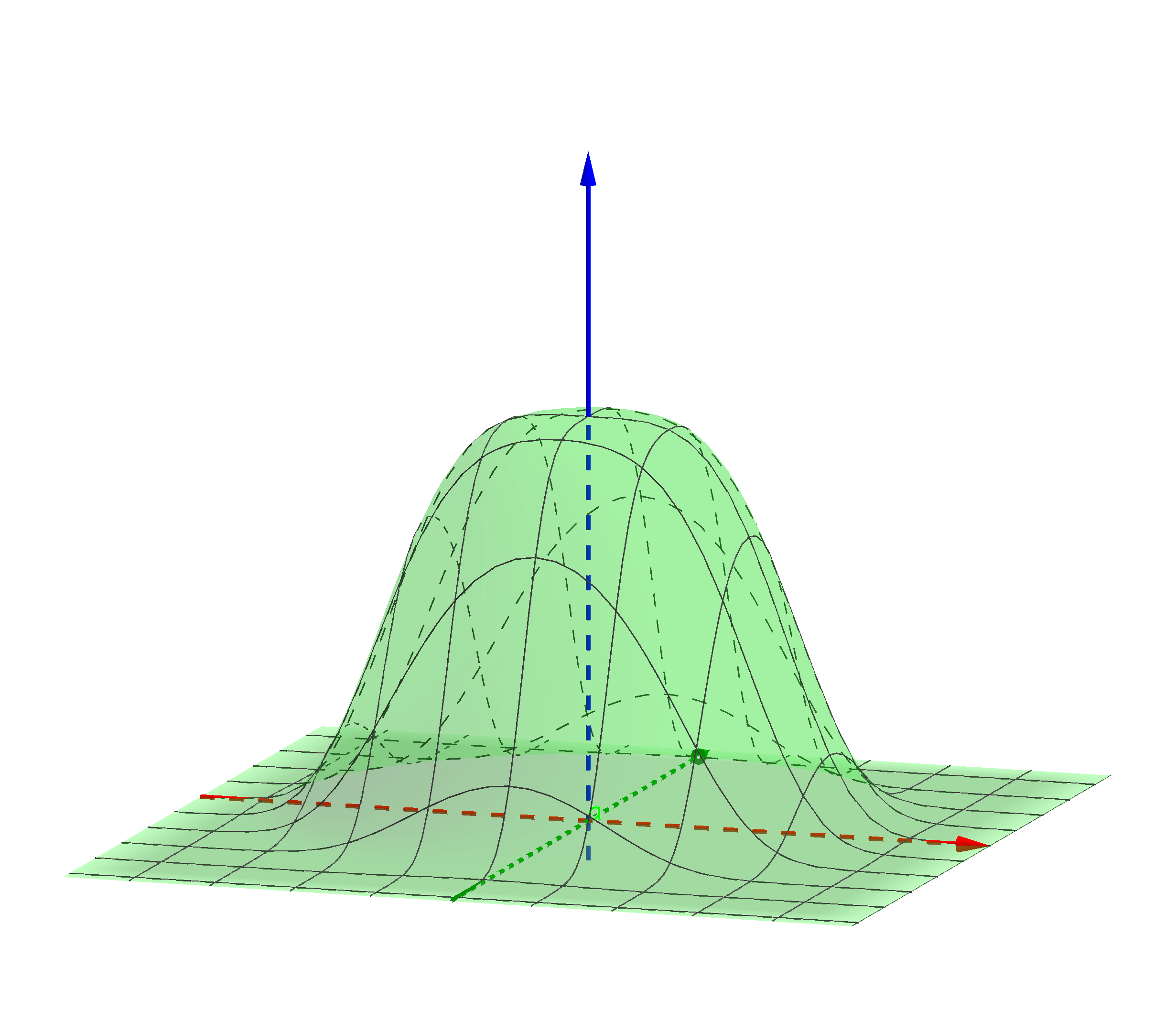}
    \subcaption{GGD with $\beta=4$.}
  \end{minipage}
  \caption{Examples of shapes of complex GGD.
  (a) When $\beta=2$, shape of GGD corresponds to that of Gaussian distribution.
  (b) When $\beta=4$, GGD is platykurtic.
  \label{fig:ggd-shape}%
  }
\end{figure}

In GGD-ILRMA,
we assume the time-variant isotropic complex GGD as the source generative model,
which is independently defined in each time-frequency slot as follows:
\begin{align}
  p(\bm Y_n) &= \prod_{i,j} p(y_{ijn}) \notag\\
  &= \prod_{i,j}
  \dfrac{\beta}{2\pi {r_{ijn}}^2 \Gamma(2/\beta)}
  \exp\p[4]{-{\dfrac{\abs{y_{ijn}}^\beta}{{r_{ijn}}^\beta}}},
  \label{eq:ggd-ilrma-generative-model}
  \\
  {r_{ijn}}^p &= \sum_k t_{ikn} v_{kjn},
  \label{eq:ggd-ilrma-source-model}
\end{align}
where the local distribution $p(y_{ijn})$ is defined as
a circularly symmetric complex Gaussian distribution,
i.e., the probability of $p(y_{ijn})$ only depends on
the power of the complex value $y_{ijn}$.
$r_{ijn}$ is the time-frequency-varying scale parameter
and $p$ is the domain parameter in NMF modeling.
Moreover,
the variables $t_{ikn}$ and $v_{kjn}$ are
the elements of the basis matrix $\bm T_n\in\rset_{\ge0}^{I\times K}$ and
the activation matrix $\bm V_n \in\rset_{\ge0}^{K\times J}$, respectively,
where $\rset_{\ge0}$ denotes the set of nonnegative real numbers.
$k = 1, \cdots, K$ is the integral index,
and $K$ is set to a much smaller value than $I$ and $J$,
which leads to the low-rank approximation.
From \eqref{eq:ggd-ilrma-generative-model},
the negative log-likelihood function $\mathcal{L}_{\mathrm{GGD}}$
of the observed signal $\bm x_{ij}$
can be obtained as follows
by assuming independence between sources:
\begin{align}
  \mathcal{L}_{\mathrm{GGD}} &= -2J\sum_i \log\abs{\det \bm W_i}\notag\\
  &\phantom{{}={}}
  +\sum_{i,j,n} \p{
    \dfrac{\abs{y_{ijn}}^\beta}{{r_{ijn}}^{\beta}} + 2\log r_{ijn}
  },
  \label{eq:ggd-ilrma-cost-function}
\end{align}
where $y_{ijn}=\bm w_{in}\bm x_{ij}$
and we used the transformation of random variables from $\bm x_{ij}$ to $\bm y_{ij}$.
The cost function of GGD-ILRMA~\eqref{eq:ggd-ilrma-cost-function}
coincides with that of IS-ILRMA when $\beta=p=2$.
By minimizing \eqref{eq:ggd-ilrma-cost-function}
w.r.t.\ $\bm W_{i}$ and $r_{ijn}$ under the limitation \eqref{eq:ggd-ilrma-source-model},
we estimate the demixing system
that maximizes the independence between sources.

Fig.\ \ref{fig:ilrma-concept} shows a conceptual model of GGD-ILRMA.
When each of the original sources has a low-rank spectrogram,
the spectrogram of their mixture should be more complicated,
where the rank of the mixture spectrogram
will be greater than that of the source spectrogram.
On the basis of this assumption, in GGD-ILRMA,
the low-rank constraint for each estimated spectrogram
is introduced by employing NMF.
The demixing matrix $\bm W_i$ is estimated
so that the spectrogram of the estimated signal
becomes a low-rank matrix modeled by $\bm T_n\bm V_n$, whose rank is at most $K$.
The estimation of $\bm W_i, \bm T_n$, and $\bm V_n$ can consistently be carried out
by minimizing \eqref{eq:ggd-ilrma-cost-function} in a fully blind manner.

\begin{figure}[tp]
\centering
\includegraphics[width=\columnwidth]{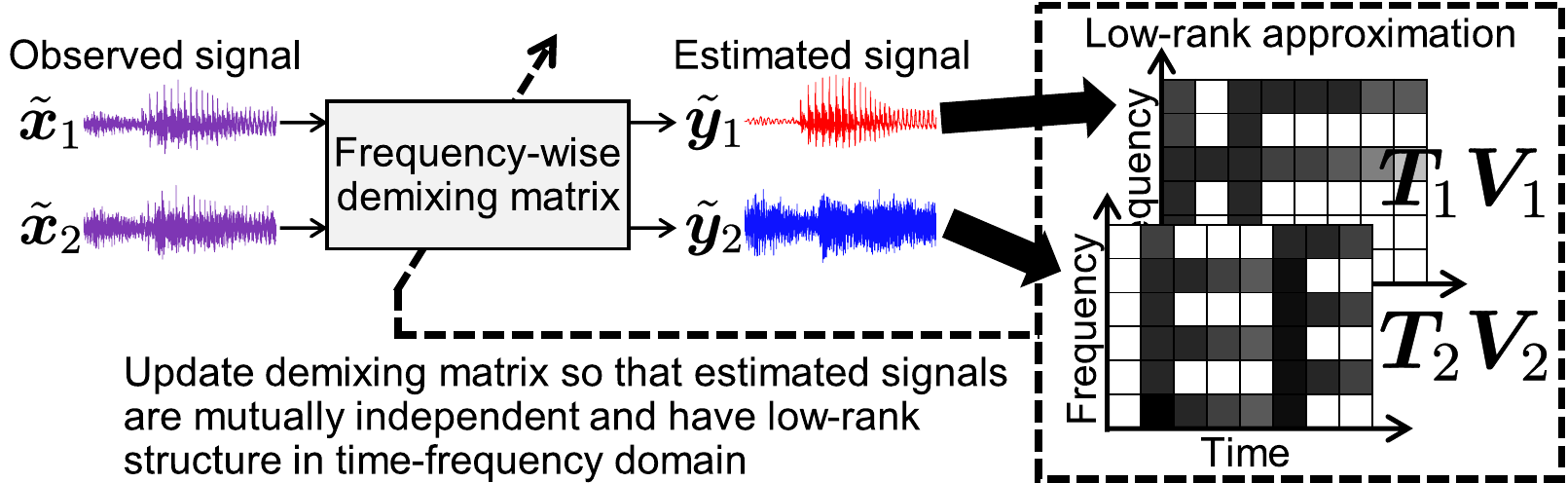}
\caption{Principle of source separation in GGD-ILRMA,
where $\tilde{\bm x}_{m}$ and $\tilde{\bm y}_{n}$ are
time-domain signals of $x_{ijn}$ and $y_{ijn}$, respectively.
\label{fig:ilrma-concept}%
}
\end{figure}


\section{Conventional Method}

\subsection{Update Rule for Demixing Matrix}
In IS-ILRMA, the demixing matrix $\bm W_i$ can be
efficiently updated by IP,
which can be applied
only when the cost function is the sum of $-\log\abs{\det \bm W_i}$
and the quadratic form of $\bm w_{in}$
(this corresponds to GGD-ILRMA with $\beta=2$).
In GGD-ILRMA, the update rule of $\bm W_i$ is also derived
using the majorization-minimization (MM) algorithm~\cite{DRHunter2000_MMalgorithm}.
When $0<\beta\le 2$,
we can use the following inequality of weighted arithmetic and geometric means
to design the majorization function:
\begin{align}
  \abs{y_{ijn}}^\beta
  \le \dfrac{\beta}{2} \dfrac{\abs{y_{ijn}}^2}{{\alpha_{ijn}}^{2-\beta}}
  + \p{1 - \dfrac\beta2} {\alpha_{ijn}}^\beta,
  \label{eq:am-gm-inequality}
\end{align}
where $\alpha_{ijn}$ is an auxiliary variable
and the equality of \eqref{eq:am-gm-inequality} holds if and only if
$\alpha_{ijn} = \abs{y_{ijn}}$.
By applying \eqref{eq:am-gm-inequality} to \eqref{eq:ggd-ilrma-cost-function},
the majorization function of \eqref{eq:ggd-ilrma-cost-function}
can be designed as
\begin{align}
  \mathcal{L}_{\mathrm{GGD}}
  &\le -2J \sum_i \log \abs{\det \bm W_i} \notag\\
  &\phantom{{}={}}
  + J \sum_{i, n} \bm w_{in}\Ht \bm F_{in} \bm w_{in} + \mathrm{const.},
  \label{eq:ggd-aux-function-less2}\\
  \bm F_{in}
  &= \dfrac{\beta}{2J}\sum_j
  \dfrac{1}{{\alpha_{ijn}}^{2-\beta} (\sum_{k} t_{ikn}v_{kjn})^{\frac\beta p}}
  \bm x_{ij} \bm x_{ij}\Ht,
  \label{eq:ggd-aux-function-less2-F}
\end{align}
where the constant term is independent of $\bm w_{in}$.
By applying IP to \eqref{eq:ggd-aux-function-less2}
and substituting the equality condition $\alpha_{ijn} = \abs{y_{ijn}}$
into \eqref{eq:ggd-aux-function-less2-F},
the update rule for $\bm W_i$ is derived as
\begin{align}
  \bm F_{in}
  &= \dfrac{\beta}{2J}\sum_j
  \dfrac{1}{\abs{y_{ijn}}^{2-\beta} (\sum_{k} t_{ikn}v_{kjn})^{\frac\beta p}}
  \bm x_{ij} \bm x_{ij}\Ht,
  \label{eq:ggd-update-w-less2-start}\\
  \bm w_{in} &\gets \bm F_{in}^{-1}\bm W_i^{-1} \bm e_n,\\
  \bm w_{in} &\gets \bm w_{in}\sqrt{1/(\bm w_{in}\Ht\bm F_{in} \bm w_{in})}.
  \label{eq:ggd-update-w-less2-end}
\end{align}
When $\beta = p = 2$, these update rules
\eqref{eq:ggd-update-w-less2-start}--\eqref{eq:ggd-update-w-less2-end}
coincide with those in IS-ILRMA.


Note that the update rules
\eqref{eq:ggd-update-w-less2-start}--\eqref{eq:ggd-update-w-less2-end}
are valid only when $0<\beta\le 2$,
which is equivalent to the condition
that the inequality \eqref{eq:am-gm-inequality} holds.
In fact, when $\beta > 2$, it is thought to be impossible to design
a majorization function to which we can apply IP
because no quadratic function w.r.t. $x$ can majorize $x^\beta$.

Conventional GGD-ILRMA achieves various types of source generative model:
when $\beta=2$, the entry of the source spectrogram follows the complex Gaussian distribution
(the same model as that of IS-ILRMA),
and when $\beta<2$, the entry of the source spectrogram follows
the complex leptokurtic distribution.
However, a source generative model that follows
a platykurtic complex GGD is yet to be achieved.
Since the marginal distribution of the time-variant super-Gaussian or Gaussian model
w.r.t. the time frame becomes only super-Gaussian,
any signals that follow sub-Gaussian distributions, such as music signals,
cannot be appropriately dealt with by the conventional GGD-ILRMA.

\subsection{Update Rule for Low-Rank Source Model%
\label{sse:conventional-w-update}}
The update rules for $\bm T_n$ and $\bm V_n$ in IS-ILRMA and GGD-ILRMA
can be derived by the MM algorithm,
which is a popular approach for NMF.
We obtain the following update rules:
\begin{align}
  t_{ikn} &\gets t_{ikn} \p{
  \dfrac{
    \beta \sum_j \dfrac{
      \abs{y_{ijn}}^\beta
      }{
      (\sum_{k'} t_{ik'n} v_{k'jn})^{\frac\beta p + 1}
    } v_{kjn}
    }{
    2 \sum_j \dfrac 1 {\sum_{k'} t_{ik'n} v_{k'jn}} v_{kjn}
  }
  }^{\frac p{\beta + p}},
  \label{eq:ggd-t-update}
  \\
  v_{kjn} &\gets v_{kjn} \p{
  \dfrac{
    \beta \sum_j \dfrac{
      \abs{y_{ijn}}^\beta
      }{
      (\sum_{k'} t_{ik'n} v_{k'jn})^{\frac\beta p + 1}
    } t_{ikn}
    }{
    2 \sum_j \dfrac 1 {\sum_{k'} t_{ik'n} v_{k'jn}} t_{ikn}
  }
  }^{\frac p{\beta + p}}.
  \label{eq:ggd-v-update}
\end{align}
See Appendix \ref{sse:update-rule-source-model}
for their detailed derivation.

In GGD-ILRMA, the cost function \eqref{eq:ggd-ilrma-cost-function}
is minimized by alternately repeating
the update of the demixing matrix $\bm W_i$ using
\eqref{eq:ggd-update-w-less2-start}--\eqref{eq:ggd-update-w-less2-end}
and the update of the low-rank source models $\bm T_n$ and $\bm V_n$
using \eqref{eq:ggd-t-update} and \eqref{eq:ggd-v-update}, respectively.
A monotonic decrease in the cost is guaranteed
over these update rules.


\newcommand{\scparam}{\eta}

\section{Proposed Method}
\subsection{Motivation\label{sse:motivation}}


The conventional methods~\cite{DKitamura2016_ILRMA,Kitamura2018_GGDILRMA}
have a limitation that
the source signal cannot be appropriately represented
when the signal follows a sub-Gaussian distribution.
In this paper,
we propose an MM-algorithm-based update rule for GGD-ILRMA
to maximize the likelihood based on the sub-Gaussian source model.
To derive the update rule, we also
extend the problem of demixing matrix estimation
into a more generalized form
and propose its optimization scheme based on GIP-HSM.

In contrast to the time-variant super-Gaussian or Gaussian model,
the marginal distribution of the time-variant sub-Gaussian model
w.r.t. the time frame
can be sub-Gaussian as well as Gaussian or super-Gaussian,
depending on its time variance of the scale parameter $r_{ijn}$.
For example, the time-variant sub-Gaussian model is platykurtic
when $r_{ijn}$ is constant w.r.t. the time frame,
whereas it becomes mesokurtic or leptokurtic
when $r_{ijn}$ fluctuates appropriately.
This shows that the proposed time-variant sub-Gaussian model covers
distributions with a wider range between platykurtic and leptokurtic shapes
than other conventional source models.
Therefore, the proposed GGD-ILRMA
is expected to have a robust performance
against the variation of the target signals.

\subsection{Derivation of GIP-HSM}
The cost functions of
IS-ILRMA and GGD-ILRMA are generalized as
\begin{align}
\mathcal{L}
&= \sum_{i=1}^I\lr[]{
-2\log\abs{\det\bm W_i}
+\sum_{n=1}^N f_{in}(\bm w_{in})}+\mathrm{const.},
\label{eq:cost-function-general}
\end{align}
where the constant term is independent of $\bm w_{in}$ and
$f_{in}\colon \cset^N\to\rset$ is a real-valued function
that satisfies the following three conditions:
\begin{enumerate}
\item $f_{in}(\bm w)$ is differentiable w.r.t.\ $\bm w$
at an arbitrary point.
\item $\Al c>0$,
$\Set{\bm w\in\cset^N}{f_{in}(\bm w)\le c}$ is convex
(naturally satisfied when $f_{in}(\bm w)$ is convex).
\item $\Al \scparam$, $f_{in}(\scparam \bm w) = \scparam^d f_{in}(\bm w)$,
namely, $f_{in}$ is a homogeneous function of degree $d$.
\end{enumerate}
The term $f_{in}(\bm w_{in})$ is determined
by the distribution of the source generative model,
e.g., $f_{in}(\bm w_{in}) = (1/J) \sum_{j}\p{{
  \abs*{\bm w_{in}\Ht \bm x_{ij}}^\beta
  }/{
  {r_{ijn}}^\beta
}}$ in GGD-ILRMA.

Here we show that
the optimization of \eqref{eq:cost-function-general} w.r.t.\ $\bm w_{in}$
is composed of ``direction optimization'' and ``scale optimization''
for each frequency bin.
Let $\bm u_{in}$ be an $N$-dimensional vector that satisfies $f_{in}(\bm u_{in})=1$.
Then, $\bm w_{in}$ can be uniquely represented as $\bm w_{in} = \scparam_{in}\bm u_{in}$,
where $\scparam_{in}$ is a positive real value.
By regarding $f_{in}(\bm u)$ as the norm of $\bm u$,
we can interpret $\bm u_{in}$ as a unit vector w.r.t. the $f_{in}$-norm.
Substituting $\bm w_{in} = \scparam_{in}\bm u_{in}$ into
\eqref{eq:cost-function-general},
the cost function is represented as
\begin{align}
\mathcal{L} &=
\sum_{i=1}^I\biggl[ -2\log\abs{\det\vtr*[\scparam_1\bm u_1,\cdots,\scparam_N\bm u_N]\Ht}\notag\\
&\hphantom{{}=\sum_{i=1}^N[]}+ \sum_{n=1}^N f_{in}(\scparam_{in} \bm u_{in}) \biggr]\notag\\
&=\sum_{i=1}^I\lr[]{
-2\log\p{\prod_n{\scparam_{in}}\cdot\abs{\det\bm U_i}}
+ \sum_{n=1}^N {\scparam_{in}}^d f_{in}(\bm u_{in})} \notag\\
&=\sum_{i=1}^I\lr[]{-2\log\abs{\det\bm U_i}
 +\sum_{n=1}^N \lr[]{-2\log{\scparam_{in}} + {\scparam_{in}}^d}},
\label{eq:cost-function-general-separated}
\end{align}
where $\bm U_{i} = \vtr*[\bm u_{i1},\cdots,\bm u_{iN}]\Ht$.
Therefore, the minimization of the cost function can be interpreted as
the minimization of $-\log\abs{\det \bm U_i}$ for each frequency bin
and the minimization of $-2\log\scparam_{in}+{\scparam_{in}}^d$
for each source and frequency bin.
These direction optimization and
scale optimization problems are independent of each other.
The optimal $\scparam_{in}$ can be calculated by a closed form because
the derivative of the cost function w.r.t.\ $\scparam_{in}$ can be written as
\begin{align}
\dif{}{\scparam_{in}}(-2\log \scparam_{in} + {\scparam_{in}}^d)
= -\dfrac{2}{\scparam_{in}} + d{\scparam_{in}}^{d-1}.
\label{eq:eta-derivative}
\end{align}
Hence,
letting the right side of \eqref{eq:eta-derivative} be zero,
we can obtain the optimal $\scparam_{in}$ as
\begin{align}
\scparam_{in} = \sqrt[d]{2/d}.
\end{align}

The actual difficulty in the optimization of the demixing matrix is
the direction optimization,
i.e.,\ the minimization of $-2\log\abs{\det \bm U_i}$.
Since minimizing $-\log x^2$ is equivalent to maximizing $x^2$,
we can reformulate this problem as
\begin{align}
  \text{maximize\ } {\abs{\det \bm U_i}^2}\quad
  \text{s.t.\ } f_{in}(\bm u_{in}) = 1.
  \label{eq:direction-opt-problem}
\end{align}
Since it is generally difficult to solve this problem by a closed form,
we apply an approach called vectorwise coordinate descent.
In this algorithm, we focus on $\bm u_{in}$, namely,
the Hermitian transpose of a particular row vector of $\bm U_i$.
By cofactor expansion,
we can deform the problem \eqref{eq:direction-opt-problem} as
\begin{align}
  \text{maximize\ } {\abs{\bm b_{in}\Ht \bm u_{in}}^2}\quad
  \text{s.t.\ } f_{in}(\bm u_{in}) = 1,
  \label{eq:objective-function-bv2}
\end{align}
where $\bm b_{in}$ is a column vector of the adjugate matrix
$\bm B_{i}=\vtr*[\bm b_{i1},\cdots,\bm b_{iN}]\Ht$ of $\bm U_{i}$.
Since $\bm b_{in}$ only depends on $\bm u_{in'}\,(n'\ne n)$
and is independent of $\bm u_{in}$,
\eqref{eq:objective-function-bv2} can be regarded as a function of $\bm u_{in}$
by fixing the other row vectors of $\bm U_i$.
Using the method of Lagrange multipliers, the stationary condition is
\begin{align}
\bm b_{in}(\bm b_{in}\Ht \bm u_{in})
+\lambda\pd{f_{in}}{\bm u_{in}\Ht}(\bm u_{in}) = 0,
\end{align}
where $\lambda$ is a Lagrange multiplier.
Since $(\bm b_{in}\Ht \bm u_{in})$ is a scalar,
the stationary condition can be rewritten as
\begin{align}
\pd{f_{in}}{\bm u_{in}\Ht} (\bm u_{in}) \parallel \bm b_{in},
\label{eq:pdfu-bin}
\end{align}
where the binary relation ``$\bm x\parallel \bm y$''
means that $\bm x$ is parallel to $\bm y$.
In \eqref{eq:pdfu-bin}, $\bm b_{in}$ is represented in terms of $\bm W_{in}$ as
\begin{align}
\bm b_{in} &= (\det\bm U_i)\bm U_i^{-1} \bm e_n \notag\\
&=(\det\bm U_i)(\mathrm{diag}(\scparam_{i1}^{-1},\ldots, \scparam_{iN}^{-1})\bm W_i)^{-1} \bm e_n \notag\\
&= (\det\bm U_i) \bm W_i^{-1}\mathrm{diag}(\scparam_{i1},\ldots, \scparam_{iN}) \bm e_n \notag\\
&= (\scparam_{in}\det\bm U_i)\bm W_i^{-1} \bm e_n \notag\\
&\parallel \bm W_{i}^{-1} \bm e_n,
\end{align}
where $\mathrm{diag}(c_1,\ldots, c_N)$ denotes
the $N\times N$ diagonal matrix whose $(n,n)$th element is $c_n$,
and $\bm e_n$ is an $N$-dimensional vector whose $n$th element is one
and whose other elements are zero.
Since $f_{in}$ is convex,
the stationary point of the objective function \eqref{eq:objective-function-bv2}
must also be the optimal point.
Therefore, the cost function~\eqref{eq:cost-function-general-separated}
that includes \eqref{eq:direction-opt-problem}
monotonically decreases with each update of the direction $\bm u_{in}$.

In conclusion, to minimize the cost function \eqref{eq:cost-function-general},
we update the vector $\bm w_{in}$ by the following two steps
in GIP-HSM.
(a) Find a vector $\bm w_{in}'$ that satisfies
\begin{align}
  \pd*{f_{in}}{\bm w_{in}\Ht} (\bm w_{in}') \parallel \bm W_{i}^{-1} \bm e_n.
  \label{eq:vcdhsm-update-a}
\end{align}
(b) Update $\bm w_{in}$ as
\begin{align}
  \npillar\bm w_{in}\gets \bm w_{in}'\sqrt[d]{2/(d\cdot f_{in}(\bm w_{in}'))}.
  \label{eq:vcdhsm-update-b}
\end{align}
The first step (a) and second step (b) correspond
to the direction and scale optimizations,
respectively.
Note that $\bm w_{in}'$ calculated in the first step
does not need to satisfy $f_{in}(\bm w_{in}') = 1$
because the scale is automatically adjusted in the second step.
In fact,
if $\bm w_{in}'$ is represented as $\bm w_{in}' = \eta_{in}'\bm u_{in}$,
the second step results in
\begin{align}
  \bm w_{in} \gets \eta_{in}'\bm u_{in}\cdot \sqrt[d]{2/(d\cdot {\eta_{in}'}^d)}
  = \bm u_{in}\sqrt[d]{2/d},
\end{align}
at which point both the direction and the scale are optimized.

\subsection{Sub-Gaussian ILRMA Based on GIP-HSM}
Using GIP-HSM,
we propose a new update rule in GGD-ILRMA
whose shape parameter $\beta$ is set to four
(time-variant sub-Gaussian model).
The cost function of GGD-ILRMA with $\beta=4$ is written as
\begin{align}
\mathcal{J} &= -2J\sum_i \log\abs{\det \bm W_i}
+\sum_{i,j,n}{
\dfrac{
  \abs*{\bm w_{in}\Ht \bm x_{ij}}^4
  }{
  {r_{ijn}}^4
}
} + \mathrm{const.},
\label{eq:ggd4-ilrma-cost-function}
\end{align}
where the constant term is independent of $\bm w_{in}$.
It seems possible
to apply GIP-HSM
by letting
$f_{in}(\bm w_{in}) = (1/J) \sum_{j}\p{{
  \abs*{\bm w_{in}\Ht \bm x_{ij}}^4
  }/{
  {r_{ijn}}^4
}}$.
In this case, however,
it is difficult to solve \eqref{eq:vcdhsm-update-a},
which is reduced to a cubic vector equation w.r.t.\ $\bm w_{in}'$.
Instead, we apply an MM algorithm to derive an update rule
that does not contain any cubic vector equations.
Hereafter, we prove the following theorem,
and then design a new type of majorization function of \eqref{eq:ggd4-ilrma-cost-function} using the theorem.

\begin{theorem}
  \label{thm:ggd4-inequality}
  Let
  $f(\bm w)=(1/J)\sum_{j=1}^J \p[1]{{\abs*{\bm w\Ht \bm x_j}^4} / {{r_j}^4}}$
  and $g(\bm w) = (\bm w\Ht \bm G \bm w)^2$,
  where $\bm G$ is defined in terms of a vector $\tilde{\bm w}$ as
  \begin{align}
    \bm H &= \vtr*[\frac1{r_{1}}\bm x_1,\cdots, \frac1{r_{J}}\bm x_J] ,\\
    \tilde{\bm q} &= \vtr*[\tilde q_1,\cdots,\tilde q_J]\T
    = \bm H\Ht \tilde{\bm w},\\
    \tilde{\bm Q} &=
    \begin{bmatrix}
      \norm{\tilde{\bm{q}}}^2&{-\tilde{q}_1\tilde{q}_2^*}&
      \cdots&{-\tilde{q}_1\tilde{q}_J^*}\\
      \rule{0pt}{3ex}{-\tilde{q}_2\tilde{q}_1^*}&\norm{\tilde{\bm{q}}}^2&
      \cdots&{-\tilde{q}_2\tilde{q}_J^*}\\
      \vdots&\vdots&\ddots&\vdots\\
      \rule{0pt}{3ex}{-\tilde{q}_J\tilde{q}_1^*}&
      {-\tilde{q}_J\tilde{q}_2^*}&
      \cdots&\norm{\tilde{\bm{q}}}^2
    \end{bmatrix},\\
    \bm G &=
    \dfrac{1}{\sqrt{J\sum_j\abs{\tilde q_j}^4}} \bm H \tilde{\bm Q} \bm H\Ht.
  \end{align}
  Then, $g(\bm w)$ satisfies $f(\bm w)\le g(\bm w)$ for arbitrary $\bm w$
  and the equality holds when $\bm w = \tilde{\bm w}$.
\end{theorem}

\begin{proof}
  Let $\bm q = \vtr*[q_1,\cdots, q_J]\T = \bm H \Ht\bm w$.
  $f(\bm w)$ and $g(\bm w)$ can be written as
  \begin{align}
    f(\bm w) &= \dfrac1J\sum_j \abs{q_j}^4,\\
    g(\bm w) &= \dfrac{1}{{J\sum_j\abs{\tilde q_j}^4}}
    \p[1]{\bm q\Ht \tilde{\bm Q} \bm q}^2.
  \end{align}
  Then, the objective inequality $f(\bm w)\le g(\bm w)$ holds if and only if
  \begin{align}
  \p[4]{\sum_j\abs{q_j}^4}\p[4]{\sum_j\abs{\tilde q_j}^4}
  \le \p[1]{\bm q\Ht \tilde{\bm Q} \bm q}^2.
  \label{eq:q4q4-qQq}
  \end{align}
  The quadratic form of the right side in \eqref{eq:q4q4-qQq} can be deformed as
  \begin{align}
    \bm q\Ht \tilde{\bm Q} \bm q
    &= \mathop{\mathrm{tr}}\p[1]{\tilde{\bm Q} \bm q\bm q\Ht} \notag\\
    &=\p[4]{\sum_j\abs{q_j}^2}\norm{\tilde{\bm q}}
    - \sum_{i\ne j} q_iq_j^* \tilde q_i^* \tilde q_j \notag\\
    &=\p[4]{\sum_j\abs{q_j}^2}\p[4]{\sum_j\abs{\tilde q_j}^2}
    - \sum_{i\ne j} q_iq_j^* \tilde q_i^* \tilde q_j.
  \end{align}
  Hence, we prove the following inequality hereafter:
  \begin{align}
    &\p[4]{\sum_j\abs{q_j}^4}\p[4]{\sum_j\abs{\tilde q_j}^4} \notag\\
    &\le\p{\p[4]{\sum_j\abs{q_j}^2}\p[4]{\sum_j\abs{\tilde q_j}^2}
    - \sum_{i\ne j} q_iq_j^* \tilde q_i^* \tilde q_j}^2.
    \label{eq:q4q4-qq2}
  \end{align}
  Let
  \begin{align}
    x_1 &= \sum_j \abs{q_j}^2,\\
    x_2 &= \sqrt{\sum_{i\neq j} \abs{q_i}^2\abs{q_j}^2},\\
    y_1 &= \sum_j \abs{\tilde q_j}^2,\\
    y_2 &= \sqrt{\sum_{i\neq j} \abs{\tilde q_i}^2\abs{\tilde q_j}^2}.
  \end{align}
  Since
  \begin{align}
    {x_1}^2 - {x_2}^2 = \sum_j \abs{q_j}^4 \ge 0,
    \label{eq:x1-x2}\\
    {y_1}^2 - {y_2}^2 = \sum_j \abs{\tilde q_j}^4 \ge 0,
    \label{eq:y1-y2}
  \end{align}
  and $x_1, x_2, y_1, y_2\ge 0$, it is obvious that
  \begin{align}
    x_1 y_1 - x_2 y_2 \ge 0.
    \label{eq:x1y1-x2y2}
  \end{align}
  Furthermore, we obtain the following inequality
  by applying the Cauchy--Schwarz inequality:
  \begin{align}
    x_2 y_2 &= \sqrt{
      \p[4]{\sum_{i\ne j}\abs{q_i}^2\abs{q_j}^2}
      \p[4]{\sum_{i\ne j}\abs{\tilde q_i}^2\abs{\tilde q_j}^2}
    } \notag\\
    &\ge \abs{
    {\sum_{i\ne j}{q_i q_j^* \tilde q_i^* \tilde q_j}}
    }
    \ge{
    {\sum_{i\ne j}{q_i q_j^* \tilde q_i^* \tilde q_j}}
    },
    \label{eq:y1y2-qiqj}
  \end{align}
  where we used the fact that
  \begin{align}
    \sum_{i\ne j}{q_i q_j^* \tilde q_i^* \tilde q_j}
    &= \sum_{i<j} \p{q_i q_j^* \tilde q_i^* \tilde q_j + q_i^* q_j\tilde q_i \tilde q_j^*} \notag\\
    &= 2\sum_{i<j} \Re[{q_i q_j^* \tilde q_i^* \tilde q_j}] \in\rset.
  \end{align}
  From \eqref{eq:x1y1-x2y2} and \eqref{eq:y1y2-qiqj},
  \begin{align}
    (x_1 y_1 - x_2 y_2)^2
    \le \p[4]{
    x_1 y_1 - {\sum_{i\ne j}{q_i q_j^* \tilde q_i^* \tilde q_j}}
    }^2.
    \label{eq:x1y1x2y2-x1y1qiqj}
  \end{align}
  Therefore, \eqref{eq:q4q4-qq2} can be proven
  by using the Brahmagupta identity,
  \eqref{eq:x1-x2}, \eqref{eq:y1-y2}, and
  \eqref{eq:x1y1x2y2-x1y1qiqj},
  as follows:
  \begin{align}
    &\p[4]{\sum_j\abs{q_j}^4}\p[4]{\sum_j\abs{\tilde q_j}^4} \notag\\
    &=({x_1}^2 - {x_2}^2)({y_1}^2 - {y_2}^2) \notag\\
    &=(x_1 y_1 - x_2 y_2)^2 - (x_1 y_2 - x_2 y_1)^2 \notag\\
    &\le (x_1 y_1 - x_2 y_2)^2 \notag\\
    &\le \p[4]{
    x_1 y_1 - {\sum_{i\ne j}{q_i q_j^* \tilde q_i^* \tilde q_j}}
    }^2 \notag\\
    &=\p{
    \p[4]{\sum_j \abs{q_j}^2}\p[4]{\sum_j \abs{\tilde q_j}^2}
    -{\sum_{i\ne j}{q_i q_j^* \tilde q_i^* \tilde q_j}}
    }^2.
  \end{align}
  It is easy to prove that the equality of \eqref{eq:q4q4-qq2} holds
  if $\bm w = \tilde{\bm w}$
  because then $\bm q = \tilde{\bm q}$ holds.
\end{proof}

Applying Theorem~\ref{thm:ggd4-inequality},
we can design a majorization function of \eqref{eq:ggd4-ilrma-cost-function}
as
\begin{align}
  \mathcal{J}&\le
  -2J\sum_i\log\abs{\det\bm W_i}
  +J\sum_{i, n} (\bm w_{in}\Ht \bm G_{in} {\bm w_{in}})^2
  +\mathrm{const.} \notag\\
  &=:\mathcal{J}^+,
  \label{eq:ggd4-ilrma-aux-function}
\end{align}
where
\begin{align}
  \bm H_{in}
  &= \vtr*[\frac1{r_{i1n}}\bm x_{i1},\cdots, \frac1{r_{iJn}}\bm x_{iJ}] ,\\
  \tilde{\bm q}_{in} &= \vtr*[\tilde q_{i1n},\cdots,\tilde q_{iJn}]\T
  = \bm H_{in}\Ht\tilde{\bm w}_{in},
  \label{eq:ggd4-ilrma-aux-function-q}\\
  \tilde{\bm Q}_{in} &=
  \begin{bmatrix}
    \norm{\tilde{\bm{q}}_{in}}^2&{-\tilde{q}_{i1n}\tilde{q}_{i2n}^*}&
    \cdots&{-\tilde{q}_{i1n}\tilde{q}_{iJn}^*}\\
    \rule{0pt}{3ex}{-\tilde{q}_{i2n}\tilde{q}_{i1n}^*}&\norm{\tilde{\bm{q}}_{in}}^2&
    \cdots&{-\tilde{q}_{i2n}\tilde{q}_{iJn}^*}\\
    \vdots&\vdots&\ddots&\vdots\\
    \rule{0pt}{3ex}{-\tilde{q}_{iJn}\tilde{q}_{i1n}^*}&
    {-\tilde{q}_{iJn}\tilde{q}_{i2n}^*}&
    \cdots&\norm{\tilde{\bm{q}}_{in}}^2\\
  \end{bmatrix},
  \label{eq:ggd4-ilrma-aux-function-Q}\\
  \bm G_{in} &=
  \dfrac{1}{\sqrt{J\sum_j\abs{\tilde q_{ijn}}^4}}
  \bm H_{in}\tilde{\bm Q}_{in} \bm H_{in}\Ht,
  \label{eq:ggd4-ilrma-aux-function-U}
\end{align}
where $\tilde{\bm w}_{in}$ is an auxiliary variable
and the equality of \eqref{eq:ggd4-ilrma-aux-function} holds
when $\bm w_{in} = \tilde{\bm w}_{in}$.
Since $g_{in}(\bm w_{in}) = (\bm w_{in}\Ht \bm G_{in} {\bm w_{in}})^2$
is a differentiable, convex, and homogeneous function of $\bm w_{in}$,
we can apply
GIP-HSM
to minimize $\mathcal{J}^+$.
The optimal condition for the direction of $\bm w_{in}$ is
determined as
\begin{align}
  \pd{g_{in}}{\bm w_{in}\Ht}(\bm w_{in}')
  = ({\bm w_{in}'}\Ht \bm G_{in} \bm w_{in}')\bm G_{in} \bm w_{in}'
  \parallel \bm W_{i}^{-1}\bm e_{n}.
  \label{eq:ggd4-direction-stationary-condition}
\end{align}
Since $({\bm w_{in}'}\Ht \bm G_{in} \bm w_{in}')$ is a scalar,
one of the solutions of \eqref{eq:ggd4-direction-stationary-condition} is
\begin{align}
  \bm w_{in}' = \bm G_{in}^{-1} \bm W_i^{-1} \bm e_n.
\end{align}
Substituting $\tilde{\bm w}_{in} = \bm w_{in}$ into
\eqref{eq:ggd4-ilrma-aux-function-q}--\eqref{eq:ggd4-ilrma-aux-function-U},
we obtain the following update rule for optimizing the direction of $\bm w_{in}$:
\begin{align}
  \bm H_{in} &= \vtr*[\frac1{r_{i1n}}\bm x_{i1},\cdots, \frac1{r_{iJn}}\bm x_{iJ}],
  \label{eq:ggd4-w-update-start}\\
  {\bm q}_{in} &= \vtr*[ q_{i1n},\cdots, q_{iJn}]\T
  = \bm H_{in} \Ht{\bm w}_{in},\\
  \bm Q_{in} &=
  \begin{bmatrix}
    \norm{{\bm{q}}_{in}}^2&{-{q}_{i1n}{q}_{i2n}^*}&
    \cdots&{-{q}_{i1n}{q}_{iJn}^*}\\
    \rule{0pt}{3ex}{-{q}_{i2n}{q}_{i1n}^*}&\norm{{\bm{q}}_{in}}^2&
    \cdots&{-{q}_{i2n}{q}_{iJn}^*}\\
    \vdots&\vdots&\ddots&\vdots\\
    \rule{0pt}{3ex}{-{q}_{iJn}{q}_{i1n}^*}&
    {-{q}_{iJn}{q}_{i2n}^*}&
    \cdots&\norm{{\bm{q}}_{in}}^2
  \end{bmatrix},
  \label{eq:ggd4-w-update-calcQ}\\
  \bm G_{in} &=
    \dfrac{1}{\sqrt{J\sum_j\abs{ q_{ijn}}^4}}
    \bm H_{in} \bm Q_{in} \bm H_{in}\Ht,
    \label{eq:ggd4-w-update-calcU}
    \\
  \bm w_{in} &\gets \bm G_{in}^{-1}\bm W_i^{-1} \bm e_n.
\end{align}
Finally, we operate the following scale optimization by applying
\eqref{eq:vcdhsm-update-b}:
\begin{align}
  {\bm q}_{in} &= \vtr*[ q_{i1n},\cdots, q_{iJn}]\T
  = \bm H_{in} \Ht{\bm w}_{in},\\
  \bm w_{in} &\gets \bm w_{in}\sqrt[4]{J/(2\textstyle\sum_j \abs{q_{ijn}}^4)},
  \label{eq:ggd4-w-update-end}
\end{align}
which is the scale optimization w.r.t.\ the $f_{in}$-norm.

In GGD-ILRMA with $\beta=4$,
the demixing matrix $\bm W_i$ is updated by
\eqref{eq:ggd4-w-update-start}--\eqref{eq:ggd4-w-update-end},
and the low-rank models $\bm T_n$ and $\bm V_n$
are updated by \eqref{eq:ggd-t-update} and \eqref{eq:ggd-v-update},
respectively.
These update rules are derived
using the MM algorithm and GIP-HSM,
thus guaranteeing a monotonic decrease of the cost function~\eqref{eq:ggd4-ilrma-cost-function}.



\section{Experimental Evaluation}

\subsection{BSS Experiment on Music Signals}
We compared the separation performance of
the proposed sub-Gaussian GGD-ILRMA ($\beta=4$)
with those of conventional
IS-ILRMA~\cite{DKitamura2016_ILRMA} and
GGD-ILRMA ($\beta<2$)~\cite{Kitamura2018_GGDILRMA}.
We artificially produced monaural dry music sources
of four melody parts
(melody 1: main melody, melody 2: counter melody, midrange, and bass)
using Microsoft GS Wavetable Synth,
where several musical instruments were chosen to play these melody parts~\cite{Kitamura2015_NMF,Kitamura_songKitamura}.
Six combinations of sources, Music 1--Music 6,
were constructed by selecting typical combinations
of instruments with different melody parts.
The combinations of dry sources
used in this experiment are shown in Table~\ref{tab:combination-drysrc}.
To simulate a reverberant mixture,
the observed signals were produced by convoluting the impulse response E2A,
which was obtained from the RWCP database~\cite{SNakamura2000_RWCP},
shown in Fig.~\ref{fig:impulse-response}.
As the evaluation score,
we used the improvement of the signal-to-distortion ratio (SDR)
~\cite{EVincent2006_BSSEval},
which indicates the overall separation quality.
An STFT was performed using a 128-ms-long Hamming window
with a 64-ms-long shift.
The shape parameter $\beta$ of the GGD was set to
1, 1.99, 2 in conventional GGD-ILRMA and 4 in the proposed method,
where $\beta=1.99$ is the best parameter according to \cite{Kitamura2018_GGDILRMA}.
The other conditions are shown in Table~\ref{tab:exp-condition}.

\begin{table}[tp]
\centering
\caption{Combinations of dry sources
\label{tab:combination-drysrc}%
}
\begin{tabular}{ccc}
\toprule
Index&Source 1& Source 2\\
\midrule
Music~1 & Fg.\ (bass)     & Ob.\ (melody 1)\\
Music~2 & Fg.\ (bass)     & Tp.\ (melody 1)\\
Music~3 & Ob.\ (melody 1) & Fl.\ (melody 2)\\
Music~4 & Pf.\ (midrange) & Ob.\ (melody 1)\\
Music~5 & Pf.\ (midrange) & Tp.\ (melody 1)\\
Music~6 & Tp.\ (melody 1) & Fl.\ (melody 2)\\
\bottomrule
\end{tabular}
\end{table}

\begin{figure}[tp]
\centering
\includegraphics[width=.7\columnwidth]{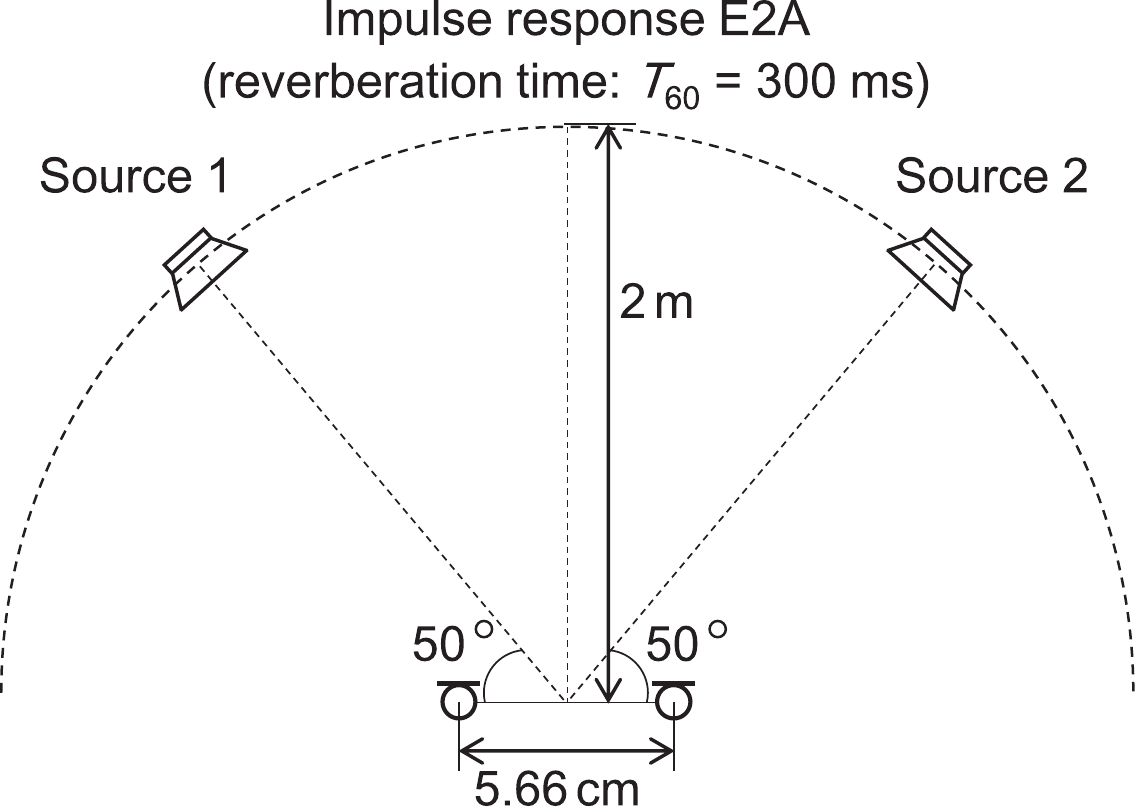}
\caption{Recording conditions of impulse response E2A ($T_{60} = \SI{300}{ms}$)
obtained from RWCP database~\cite{SNakamura2000_RWCP}.
\label{fig:impulse-response}%
}
\end{figure}

\begin{table}[tp]
\centering
\caption{Experimental conditions for music and speech source separation
\label{tab:exp-condition}%
}
\begin{tabular}{cc}
\toprule
Sampling frequency & $\SI{16}{kHz}$\\
Number of iterations& 1000\\
Number of bases& 20\\
Number of trials& 10\\
Domain parameter & $p=0.5$\\
Initial demixing matrix $\bm W_i$ & identity matrix\\
\fpillar \parbox{12em}{\centering
Entries of initial source model matrices $\bm T_n$ and $\bm V_n$
}& uniformly distributed random values\\
\bottomrule
\end{tabular}
\end{table}

Fig.~\ref{fig:result-music} shows the average SDR improvements
for Music 1--Music 6.
The proposed sub-Gaussian GGD-ILRMA on average outperforms the other conventional ILRMAs.
This result confirms that
the proposed sub-Gaussian source model is
more appropriate for dealing with music signals
than other conventional (super-)Gaussian source models.

\begin{figure}[tp]
\centering
\includegraphics[width=\columnwidth]{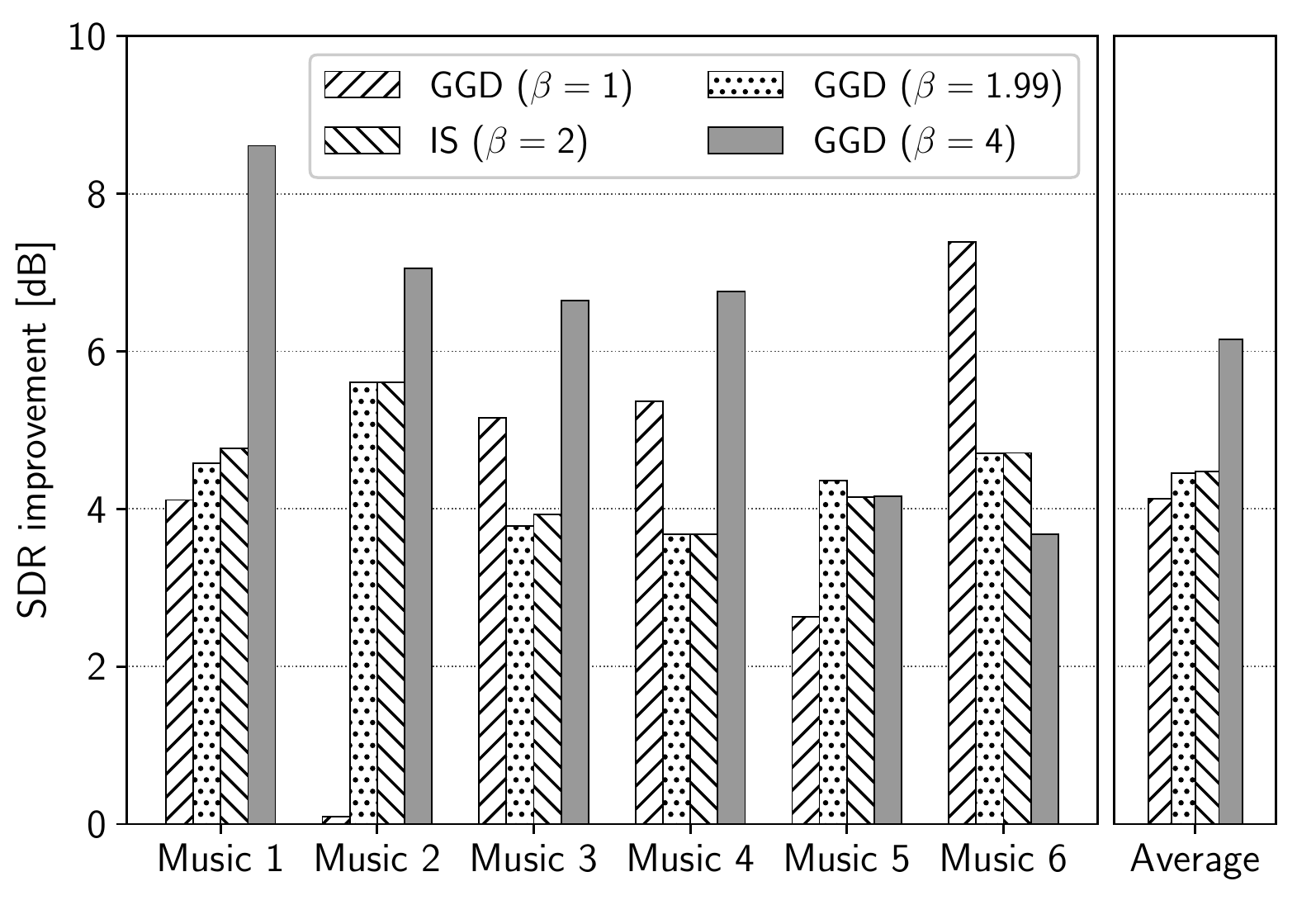}
\caption{Average SDR improvement of
GGD-ILRMA ($\beta=1$),
GGD-ILRMA ($\beta=1.99$),
IS-ILRMA,
and proposed GGD-ILRMA ($\beta=4$)
for six music signals.
``Average'' bar graph on right side shows average SDR improvement
among six music signals for each method.
\label{fig:result-music}%
}
\end{figure}

\subsection{BSS Experiment on Speech Signals}
We also confirmed
the separation performance of the proposed sub-Gaussian GGD-ILRMA
for speech signals,
which are less likely to follow a sub-Gaussian distribution
than music signals.
We used the monaural dry speech sources from the source separation task
in SiSEC2011~\cite{SAraki2012_SiSEC}, Speech~1--Speech~4.
An STFT was performed using a 256-ms-long Hamming window
with a 128-ms-long shift.
The other conditions were the same
as those of the music source separation experiment.

Fig.~\ref{fig:result-speech} shows the average SDR improvements
for Speech 1--Speech 4.
The proposed GGD-ILRMA on average outperforms the other conventional ILRMAs
even for speech signals,
which are expected to be sparse and follow super-Gaussian distributions.
This shows that the proposed time-variant sub-Gaussian model
can appropriately model super-Gaussian signals as well as sub-Gaussian signals
owing to its time-variant property, as described in Sect.~\ref{sse:motivation}.

\begin{figure}[tp]
\centering
\includegraphics[width=\columnwidth]{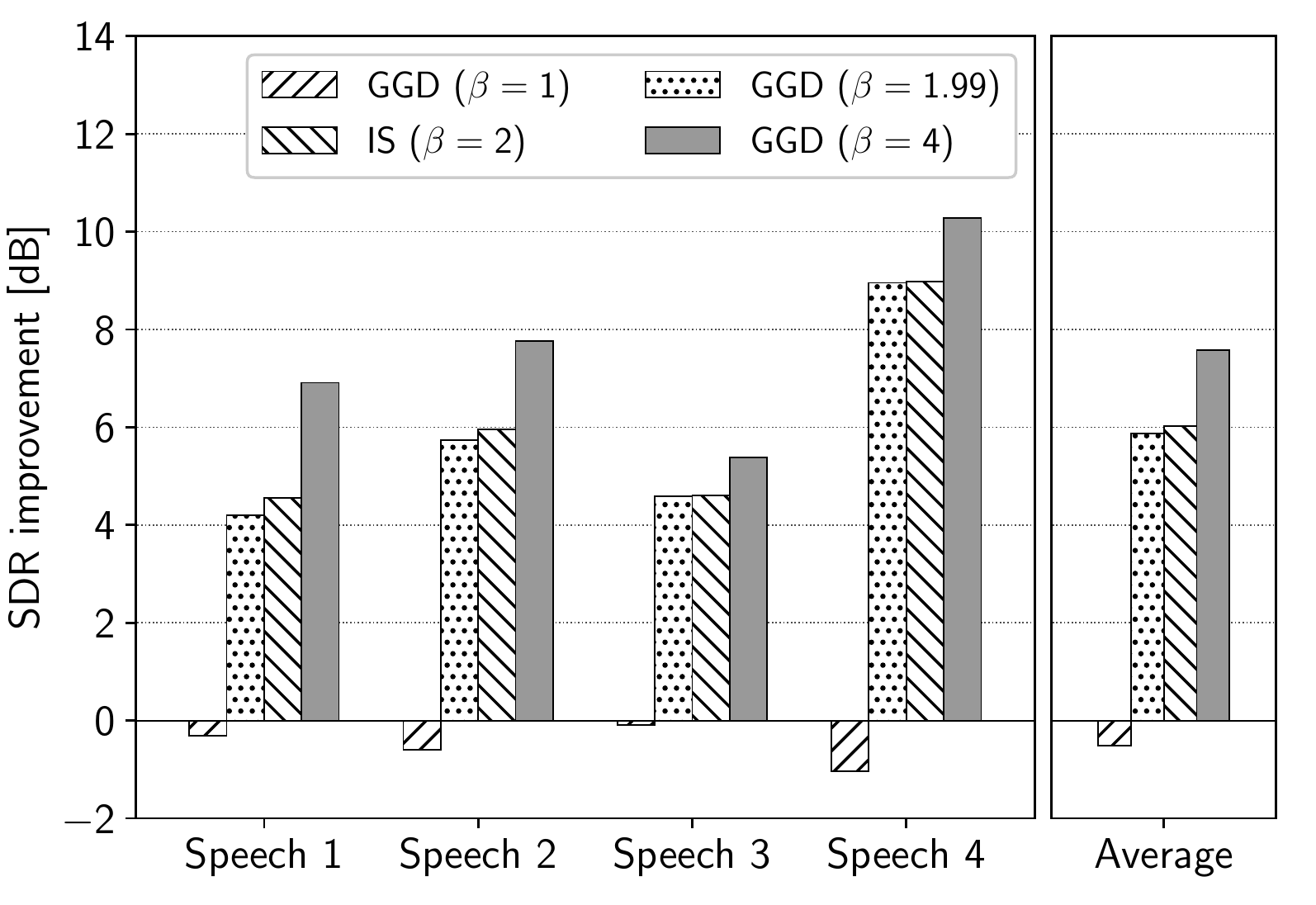}
\caption{Average SDR improvement of
GGD-ILRMA ($\beta=1$),
GGD-ILRMA ($\beta=1.99$),
IS-ILRMA,
and proposed GGD-ILRMA ($\beta=4$)
for four speech signals.
``Average'' bar graph on right side shows average SDR improvement
among four speech signals for each method.
\label{fig:result-speech}%
}
\end{figure}


\section{Conclusion}
We proposed a new type of ILRMA,
which assumes that the source signal
follows the time-variant isotropic complex sub-Gaussian GGD.
By using a new update scheme called
GIP-HSM,
we obtained a convergence-guaranteed update rule for the demixing matrix.
Furthermore, in the experimental evaluation,
we revealed the versatility of the proposed method,
i.e., the proposed time-variant sub-Gaussian source model
can deal with various types of source signal,
ranging from sub-Gaussian music signals to super-Gaussian speech signals.

\section*{Acknowledgment}
This work was partly supported by
SECOM Science and Technology Foundation
and JSPS KAKENHI Grant Numbers
JP17H06572 
and JP16H01735. 


%
%
%
%
%
%
%
%
\bibliographystyle{IEEEtran}
\bibliography{reference}


\appendix
\subsection{Derivation of Update Rule for Low-Rank Source Model%
\label{sse:update-rule-source-model}}

The update rules for $\bm T_n$ and $\bm V_n$ in GGD-ILRMA
can be derived by the MM algorithm.
In the MM algorithm, we minimize the majorization function
instead of the original cost function.
To derive the majorization function in GGD-ILRMA,
we introduce Jensen's inequality
\begin{align}
  \p{\sum_k t_{ikn}v_{kjn}}^{-\frac{\beta}{p}}
  \le \sum_k \phi_{ijnk}
  \p{\dfrac{t_{ikn}v_{kjn}}{\phi_{ijnk}}}^{-\frac{\beta}{p}}
  \label{eq:jensen-inequality}
\end{align}
and the tangent-line inequality
\begin{align}
  \log\sum_k t_{ikn}v_{kjn}
  \le \dfrac1{\psi_{ijn}}\p{\sum_k t_{ikn}v_{kjn} - 1} + \log \psi_{ijn},
  \label{eq:tangent-inequality}
\end{align}
where $\phi_{ijnk}>0$ and $\psi_{ijn}>0$
are auxiliary variables and $\phi_{ijnk}$ satisfies $\sum_k \phi_{ijnk}=1$.
The equalities of \eqref{eq:jensen-inequality} and \eqref{eq:tangent-inequality}
hold if and only if
\begin{align}
  \phi_{ijnk} &= \dfrac{
    t_{ikn} v_{kjn}
    }{
    \sum_{k'}t_{ik'n} v_{k'jn}
  },\\
  \psi_{ijn} &= \sum_k t_{ikn} v_{kjn},
\end{align}
respectively.
By substituting \eqref{eq:ggd-ilrma-source-model}
into \eqref{eq:ggd-ilrma-cost-function}
and applying \eqref{eq:jensen-inequality} and \eqref{eq:tangent-inequality}
to \eqref{eq:ggd-ilrma-cost-function},
the majorization function of \eqref{eq:ggd-ilrma-cost-function}
can be designed as
\begin{align}
  \mathcal{L}_{\mathrm{GGD}}
  &\le \sum_{i,j,n,k}\p{
    \dfrac{
      \phi_{ijnk}^{\frac\beta p + 1} \abs{y_{ijn}}^\beta
      }{
      (t_{ikn}v_{kjn})^{\frac\beta p}
    }
    + \dfrac{2t_{ikn} v_{kjn}}{p \psi_{ijn}}
  }
  + \mathrm{const.},
  \label{eq:ggd-tv-update-aux-function}
\end{align}
where the constant term is independent of $t_{ikn}$ and $v_{kjn}$.
By setting the partial derivatives of \eqref{eq:ggd-tv-update-aux-function}
w.r.t. $t_{ikn}$ and $v_{kjn}$ to zero,
we obtain the update rules \eqref{eq:ggd-t-update} and \eqref{eq:ggd-v-update},
respectively.


\end{document}